\documentclass[a4paper,11pt]{article}

\usepackage{authblk}
\usepackage[utf8]{inputenc}
\usepackage{amsthm}

\usepackage{amsmath,amssymb,enumitem}
\usepackage[linesnumbered, ruled, vlined]{algorithm2e} 

\newtheorem{theorem}{Theorem}

\newtheorem{lemma}[theorem]{Lemma}
\newtheorem{corollary}[theorem]{Corollary}
\newtheorem{definition}[theorem]{Definition}
\newtheorem{remark}{Remark}

\usepackage{tikz, pgfplots}
\usetikzlibrary{patterns}
\usetikzlibrary{shapes.geometric, calc}

\usepackage[margin=1in]{geometry}
\usepackage{hyperref}
\usepackage{color}
\definecolor{darkgreen}{RGB}{0,100,0}
\definecolor{firebrick}{RGB}{178,34,34}

\DeclareMathOperator{\Expected}{\mathbb{E}}
\newcommand{\Ex}[1]{\Expected\pbrcx{#1}}

\newcommand{\R}{{\mathbb{R}}}

\newcommand{\pth}[1]{\ensuremath{\left(#1\right)}}
\newcommand{\pbrcx}[1]{\ensuremath{\left[#1\right]}}

\newcommand{\calH}{{\cal H}}
\newcommand{\calM}{{\cal M}}
\newcommand{\calN}{{\cal N}}

\newcommand{\rad}{{\rm{rad}}}

\newcommand{\e}{{\varepsilon}}

\newcommand{\Cech}{{\v{C}ech }}

\newcommand{\conv}{{\rm conv}}

\renewcommand{\wp}{\hat{P}}

\title{Dimensionality Reduction for Persistent Homology with Gaussian Kernels}
\date{}

\author[1]{Jean-Daniel Boissonnat  
\thanks{
Supported by the French government, through the 3IA C\^ote d'Azur Investments in the Future project managed by the 
National Research Agency (ANR) with the reference number ANR-19-P3IA-0002.}
}

\author[2]{Kunal Dutta
\thanks{
Supported by the Polish NCN-SONATA Grant no. 2019/35/D/ST6/04525 
\emph{(Probabilistic tools for high-dimensional geometric inference, topological data analysis and large-scale networks)}.}
}

\affil[1]{Universit\'e C\^ote d'Azur, INRIA, Sophia-Antipolis, France, email: jean-daniel.boissonnat@inria.fr}
\affil[2]{Faculty of Mathematics, Informatics, and Mechanics, University of Warsaw, Poland. email: K.Dutta@mimuw.edu.pl}

\begin{document}

\maketitle

\begin{abstract}
 Computing persistent homology using Gaussian kernels is useful in the domains of topological data analysis and machine learning 
 as shown by Phillips, Wang and Zheng~[SoCG 2015]. However, contrary to the case of computing persistent homology using the Euclidean 
 distance or even the $k$-distance, it is not known how to compute the persistent homology of high dimensional data using Gaussian kernels. 
 In this paper, we consider a power distance 
 version of the Gaussian kernel distance (GKPD) given by Phillips, Wang and Zheng, and show  that the 
  persistent homology of the \Cech filtration of $P$ computed using the GKPD 
  is approximately preserved. For datasets in $\R^D$, under 
  a relative error bound of $\e\in (0,1]$, we obtain a dimensionality of $(i)$ $O(\e^{-2}\log^2 n)$ for $n$-point datasets and 
  $(ii)$ $O(D\e^{-2}\log (Dr/\e))$ for datasets having diameter $r$ (up to a scaling factor).  \\

  We use two main ingredients. The first one
   is a new decomposition of the squared radii of \Cech simplices using the kernel power distance, in terms of the 
  pairwise GKPDs between the vertices, which we state and prove. The second one is the Random Fourier Features (RFF) map of Rahimi 
  and Recht~[NeurIPS 2007], as used by Chen and Phillips~[ALT 2017].
\end{abstract}

        \section{Introduction}
        \label{sec:intro}

Persistent homology (PH) is one of the main tools to extract information from data in topological data analysis. Given a data set as a point cloud in some ambient 
space, the idea is to construct a filtration sequence of topological spaces from the point cloud, and extract topological information from this sequence. 

Two main issues are to be faced. First, the data points often live in a very high dimensional space and computing PH has 
exponential or worse dependence on the ambient dimension. It follows that PH rapidly becomes unusable once the dimension grows beyond a few dozens -- 
which is indeed the case in many applications, for example in image processing, neuro-biological networks, and data mining (see e.g. Giraud~\cite{giraud2014introduction}). 
This phenomenon is often referred to as the curse of dimensionality.  A second major difficulty comes from the fact that data is usually corrupted by noise
and outliers. Indeed, while the PH (computed using offsets to a distance function) is quite robust to Hausdorff noise, it is not hard to see that 
the presence of even a single outlier can significantly affect the PH (see e.g.~\cite{DBLP:journals/focm/ChazalCM11}). 

\paragraph*{Persistent Homology beyond the Euclidean distance} 
One approach to circumvent the issue of outliers is to use distance functions that are more robust to outliers, such as the distance-to-a-measure (DTM) and the related $k$-distance 
(for finite data sets), proposed recently by Chazal et al.~\cite{DBLP:journals/comgeo/BuchetCOS16}. This approach also has the advantage of compatibility with de-noising techniques 
such as~\cite{DBLP:journals/jocg/BuchetDWW18}. However, although DTM is a promising direction, an exact implementation can have significant cost in 
run-time. To overcome this difficulty,  approximations of the $k$-distance have been proposed recently that led to certified (although rather poor) approximations 
of PH (Guibas et al.~\cite{DBLP:journals/dcg/GuibasMM13}; Buchet et al.~\cite{DBLP:journals/comgeo/BuchetCOS16}). Recently, improved approximations have 
been reported using this approach~\cite{DBLP:conf/compgeom/AnaiCGIITU19}. 

Another approach to circumvent the issue of outliers involves using kernels (Phillips et al.~\cite{DBLP:conf/compgeom/PhillipsWZ15}). 
A kernel is a similarity function on pairs of points in an ambient 
space. Kernel methods are a mainstay of machine learning and data analysis.  One important reason for their popularity, is the so-called \emph{kernel trick}, 
whose underlying idea is to use a mapping $\phi$ that maps the data to some (usually infinite dimensional) Hilbert space $\calH$ where the kernel function of a pair of points $x,y\in \R^D$ 
is equal to the inner product of their images $\langle \phi (x), \phi(y)\rangle_{\calH}$. 
As explained in~\cite{DBLP:conf/compgeom/PhillipsWZ15}, computing the PH 
using the kernel distance has certain advantages compared to the Euclidean or the $k$-distance, especially for machine learning applications. These include the existence of 
$\e$-coresets~\cite{DBLP:conf/compgeom/JoshiKPV11,DBLP:conf/soda/Phillips13} as well 
as some properties of the kernel distance function, e.g. its sublevel sets map to the superlevel sets of kernel density estimates (commonly used in machine learning), 
it is Lipschitz with respect to a smoothing parameter $\sigma$ which can be varied for a fixed input, and the kernel distance between two measures is bounded 
by the Wasserstein $2$-distance when $\sigma$ goes to infinity. These properties of the Gaussian kernel distance were proved in~\cite{DBLP:conf/compgeom/PhillipsWZ15}, 
where an approximate power distance version of the kernel distance -- which we call the Gaussian Kernel Power Distance (GKPD) -- was used to compute the PH of some 
datasets and compare with the PH computed using existing 
distance functions. Further progress in constructing robust persistence diagrams has been recently made in~\cite{10.5555/3495724.3497561} using this approach.

\paragraph*{Dimension Reduction and Persistent Homology} Coming to the problem of high dimensionality, one of the simplest and most commonly used mechanisms 
to mitigate the curse of dimensionality is using random projections, as applied in the celebrated Johnson-Lindenstrauss Lemma~\cite{johnson_lindenstrauss_1984}
(JL Lemma for short). The JL Lemma states 
that for any $\e\in(0,1)$, any set of $n$ points in Euclidean space can be embedded into a space of dimension $O(\e^{-2} \log n)$ with $(1 \pm \e)$ distortion. 
Since the initial non-constructive proof of this fact by Johnson and Lindenstrauss (1984), several authors have given successive improvements, 
e.g.~\cite{DBLP:conf/stoc/IndykMRV97,DBLP:journals/rsa/DasguptaG03,ACHLIOPTAS2003671,Ailon:2006:ANN:1132516.1132597,doi:10.1002/rsa.20218,DBLP:journals/jacm/KaneN14},
addressing the issues of efficient construction and implementation, using sparse random matrices that support fast multiplication.
Recently, Narayanan and Nelson~\cite{DBLP:conf/stoc/NarayananN19}, building on~\cite{DBLP:journals/tcs/ElkinFN17,DBLP:conf/stoc/MahabadiMMR18}, 
showed that for a given set of points or terminals, using a \emph{non-linear} mapping, it is possible to achieve dimensionality 
reduction while preserving distances from any terminal to any point in the ambient space. \\

The JL Lemma has also been used by Sheehy~\cite{DBLP:conf/compgeom/Sheehy14} and Lotz~\cite{doi:10.1098/rspa.2019.0081} to reduce the complexity of computing 
PH for point sets of bounded cardinality or Gaussian width. 
Lotz's result also implies dimensionality reductions for sets of 
bounded doubling dimension, in terms of the spread (ratio of the maximum to minimum interpoint distance). However, their techniques involve only the usual distance to a 
point set and therefore are highly sensitive to the presence of outliers and noise as mentioned earlier.
The question of adapting the 
method of random projections in order to reduce the complexity of 
computing PH using the $k$-distance is therefore a natural one, and was addressed by Arya et al. in~\cite{DBLP:journals/jact/AryaBDL21} 
who showed that under random projections, the same bounds apply for the preservation of PH using the $k$-distance, as for pairwise distances. 

\paragraph*{Dimensionality Reduction for Persistent Homology using Kernel Distance} In computations involving the kernel distance, as discussed earlier, 
there is a plethora of techniques that address the issue of large data size. However, especially for computing the PH, 
there are not many techniques known to address the question of large ambient dimensionality. For instance the JL Lemma doesn't work in this case because 
the kernel distance function involves higher powers of the squared Euclidean distance. The JL Lemma does allow the construction of 
an approximate embedding of the kernel power distance as a Euclidean distance. However, such a construction would involve an $n\times n$ Cholesky decomposition, as well as  
other computationally expensive procedures. Further, the embedding would become dependent on the input data points and therefore would no longer be 
data-oblivious.\footnote{Probabilistic dimensionality reduction techniques are not completely independent of the data, since the success of the algorithm needs to be
verified by comparing the original and projected data points.}

In another direction, a mapping given by Chen and Phillips~\cite{DBLP:conf/alt/ChenP17}  
using the Random Fourier Features (RFF) of Rahimi and Recht~\cite{DBLP:conf/nips/RahimiR07}, preserves pairwise kernel distances up to a $(1\pm \e)$-factor. 
However it does not preserve distances between measures, which a priori seems necessary for computing the PH under the framework of~\cite{DBLP:conf/compgeom/PhillipsWZ15}. 
Similarly, other different approaches (e.g.~\cite{10.5555/3322706.3322718,10.5555/3294996.3295197,10.5555/3305381.3305408,8104131}) 
are not efficient in preserving distances between point sets, or between general measures.
Another embedding obtained by Phillips and Tai~\cite{DBLP:conf/approx/PhillipsT20} gives a relative 
approximation with a small additive error for kernel distances between sets of points.
Even given such a mapping (i.e. preserving kernel distances between point distributions), it is not clear that it can 
preserve the PH, since this involves preserving intersections of multiple balls under a power distance 
(see e.g.~\cite{DBLP:journals/jact/AryaBDL21,DBLP:conf/compgeom/PhillipsWZ15}). A key issue under the GKPD is that  
the weights associated to the data points are not just a function of the points themselves, but of the pairwise kernel distances of 
all the points in the data set. This means that given any mapping of the data points to a lower dimensional space, the weights of the points must be 
recomputed in the new space. A similar problem arises in the case of computing the persistent homology under the $k$-distance. Addressing this 
was asked as an open problem in~\cite{DBLP:conf/compgeom/Sheehy14}, and recently answered by Arya et al.~\cite{DBLP:journals/jact/AryaBDL21}. 
However, their solution crucially uses the linearity of the dimensionality-reducing map as well as certain properties of minimum enclosing balls under the $k$-distance, 
which they prove, and hence does not apply to the GKPD. 
As mentioned in the conclusion of~\cite{DBLP:journals/jact/AryaBDL21}, 
it is possible to obtain a constant factor approximation of the PH, essentially by approximating the kernel distance by the Euclidean distance for small values
of the Euclidean distance. 
However, the question of finding a $(1+\e)$-factor approximation for the GKPD remained open.

\subsection{Our Contribution}
In this paper, we show that given any $\e\in (0,1]$, it is possible to approximate the PH of an $n$-point dataset, computed using the power distance version 
of the Gaussian kernel distance 
(used by Phillips et al.~\cite{DBLP:conf/compgeom/PhillipsWZ15}), by a $(1+\e)$-factor 
while reducing the dimensionality down to $O\pth{\e^{-2}\log n}$.
Our results are analogous to the ones of Sheehy~\cite{DBLP:conf/compgeom/Sheehy14} and Lotz~\cite{doi:10.1098/rspa.2019.0081} 
for PH computed using Euclidean distances, and Arya et al.~\cite{DBLP:journals/jact/AryaBDL21} for the Euclidean $k$-distance. 
Further, since our target dimension is of the same order as the target dimension of the Johnson Lindenstrauss bounds, which 
are optimal for the Euclidean distance up to constant factors, and the Gaussian kernel distance is a Lipschitz function of the Euclidean distance, 
this implies that our target dimension is also optimal (up to constant factors).  

Informally, our main theorem states that 
   given a set of points in a high-dimensional space, there exists an efficiently computable  mapping of the points onto a lower dimensional target space, 
which approximates their persistent homology computed using a Gaussian kernel, to an arbitrary degree. Further, the dimensionality of the target space only 
depends inverse quadratically on the desired accuracy parameter and logarithmically on the $(i)$ number of points or $(ii)$ the diameter of the dataset.  

        The formal version, Theorem~\ref{thm:gauss-kern-dimred}, is in Section~\ref{sec:main-result-proof}. It yields a map  
        allowing us to approximately compute the PH of a set of points in a high dimensional space, using the GKPD, while actually working with Euclidean distances 
        in a lower-dimensional space. Thus, it affirmatively answers the question asked in the conclusion of~\cite{DBLP:journals/jact/AryaBDL21}. 

        Our main tool, is a new decomposition theorem showing that the squared radius of a minimum enclosing ball of a set of weighted points, computed using 
        the GKPD, can be expressed as a linear combination of pairwise power distances between the points. We shall show that although the 
        GKPD is non-linear, when lifted to a certain Hilbert space, it has several nice properties, which we then use to prove the decomposition theorem.
         


     \paragraph*{Organization of paper} The rest of this paper is organized as follows. In Section~\ref{sec:backgd} we provide some necessary background and preliminary details.
     In Section~\ref{sec:mepb} we prove some properties of minimum enclosing balls of weighted points in a Hilbert space. In Section~\ref{sec:low-dist-map-pd} we 
     study the stability of \Cech filtrations constructed using the GKPD, under low-distortion maps, proving Theorem~\ref{thm:inner-prod-dimred-gauss-kern}. 
     In Section~\ref{sec:main-result-proof} 
     we prove our main theorem, showing how the \emph{Random Fourier Features} map of Rahimi and Recht, together with our new decomposition result (Theorem~\ref{thm:inner-prod-dimred-gauss-kern}) 
     gives the proof of Theorem~\ref{thm:gauss-kern-dimred}. 
     We conclude with a few remarks and open questions in Section~\ref{sec:concl}.

        \section{Background}
        \label{sec:backgd}

        We briefly introduce some of the definitions and tools needed for our results and proofs. For a deeper picture, the 
        references~\cite{DBLP:journals/comgeo/BuchetCOS16,DBLP:journals/focm/ChazalCM11} 
        would be greatly beneficial to the reader. We also refer the interested reader to~\cite{DBLP:journals/jact/AryaBDL21,DBLP:conf/compgeom/PhillipsWZ15} for further reading.

        \subsection{Persistent Homology}
        \label{sec:persh}
\noindent {\bf Simplicial Complexes and Filtrations} 
Let $V$ be a finite set. 
      An (abstract) simplicial complex with vertex set $V$ is a set $ K $ of finite
      subsets of $V$  such that if $ A \in K $ and $ B \subseteq A$,
then $ B \in K $. The sets in $ K $ are called
      the simplices of $ K $. A simplex $F \in K$ that is strictly contained in a simplex 
      $A\in K$, is said to be a \emph{face} of $A$.

	A simplicial complex $ K $ with a function $ f: K \to
        \mathbb{R} $ such that $ f(\sigma) \le f(\tau) $ whenever
        $\sigma$ is a face of $\tau$ is a filtered simplicial
        complex. The sublevel set of $f$ at $ r \in \mathbb{R}$, $
        \mathnormal{f}^{-1}\left(-\infty,r \right] $, is a subcomplex of $ K $. 
	By considering different values of $ r $, we get a nested
        sequence of subcomplexes (called a filtration) of $ K $, $ \emptyset= K^0\subseteq K^1 \subseteq ... \subseteq K^m=K $, where $ K^{i} $ is the sublevel set at value $ r_i $. 
The        \Cech filtration associated to  a finite set $P$ of points
in $\mathbb{R}^D $ plays an important role in Topological Data Analysis.

	\begin{definition}[\Cech Complex]
 The \Cech complex $\check{C}_\alpha(P)$ is the set of simplices $\sigma\subset P$ such that rad($\sigma$) $\le$ $\alpha$, where $\rad(\sigma)$ is the radius of the smallest 
	enclosing ball of $ \sigma $, i.e.
$$\rad(\sigma) \leq \alpha \Leftrightarrow \exists x\in \R^D,\; \forall p_i\in \sigma, \; \|x-p_i\| \leq \alpha.$$
\end{definition}
When $\alpha$ goes from $0$ to $+\infty$, we obtain the \Cech
filtration of $P$.
$\check{C}_\alpha(P)$ can be equivalently defined as the \emph{nerve}  of the
closed balls
          $\overline{B}(p,\alpha)$, centered at the points in $P$ and
          of radius $\alpha$:  \[ \check{C}_\alpha(P)
          = \{ \sigma \subset P | \cap_{p \in
            \sigma}\overline{B}(p,\alpha) \neq \emptyset \}. \] 
	By the Nerve Lemma (e.g.~\cite{ghrist2014elementary,Borsuk1948}), we know that the union of balls
        $U_\alpha =\cup_{p\in P} \overline{B}(p,\alpha) $, $p\in P$,
	and $ \check{C}_\alpha(P) $ have the same homotopy type. 
Moreover, since the union of balls of a good sample $P$ of a reasonably regular shape $X$ captures the homotopy type of $X$, 
computing the \Cech complex of $P$ will provide the homotopy type of $X$. We also recall that a simpler complex called the 
{\em $\alpha$-complex} of $P$ (see e.g.~\cite{DBLP:books/daglib/0025666}) captures also the homotopy type of 
$\bigcup_{p\in P} \overline{B}(p,\alpha)$. Our results will apply to both complexes.

\vspace{2mm}

\noindent {\bf Persistence Diagrams.} 
Persistent homology is a means to compute and record the changes in the topology of the filtered complexes as 
the parameter $\alpha$ increases from zero to infinity. Edelsbrunner, Letscher and Zomorodian~\cite{DBLP:journals/dcg/EdelsbrunnerLZ02} 
gave an algorithm to compute the PH, which takes a filtered simplicial complex as input, and 
outputs a sequence $(\alpha_{birth},\alpha_{death})$ of pairs of real numbers. Each such pair corresponds to a topological feature, and records the values 
of $\alpha$ at which the feature appears and disappears, respectively, in the filtration. Thus the topological features 
of the filtration can be represented using this sequence of pairs, which can be represented either as points in the 
extended plane $\bar{\R}^2 = \pth{\R\cup \{-\infty,\infty\}}^2$, 
called the \emph{persistence diagram} or as a sequence of barcodes (the \emph{persistence barcode}) (see, e.g.,~\cite{DBLP:books/daglib/0025666}).
A pair of persistence diagrams $\mathbb{G}$ and $\mathbb{H}$ corresponding to the filtrations $(G_\alpha)$ and $(H_\alpha)$ respectively,
are \emph{multiplicatively $\beta$-interleaved}, $(\beta \geq 1)$, if for all $\alpha$, we have that
$G_{\alpha/\beta}  \subseteq H_{\alpha} \subseteq G_{\alpha\beta}$. We shall crucially rely on the fact that a given 
persistence diagram is closely approximated by another one if they are multiplicatively $c$-interleaved, with $c$ close to $1$ 
(see e.g.~\cite{DBLP:books/daglib/0039900}).

	The Persistent Nerve Lemma \cite{DBLP:conf/compgeom/ChazalO08} shows that the PH of the \Cech filtration is the 
	same as the homology of the sublevel filtrations of the distance function. The same result also holds for the 
        Delaunay filtration~\cite{DBLP:conf/compgeom/ChazalO08}.

\vspace{2mm}

\subsection{\bf Distance to Measure; PH with Power Distances}
\label{par:pers-homol-pow-dist}
 The most common approach in Topological Data Analysis is to consider the distance function given by the shortest distance to 
a point in $V$, i.e. $d_V:\R^D \to \R_+$ is $d_V(x)= \inf_{y\in V} d(x,y)$. Given this distance function one can construct 
the \Cech filtration by considering the $\alpha$-offsets of $d_V(.)$ (i.e. the sublevel sets $\{x\in \R^D \;\;|\;\; d_V(x)\leq \alpha\}$)
as unions of balls, and computing the nerve of these unions.
However, as mentioned in the Introduction, the PH obtained from this choice of distance function is highly sensitive 
to outliers, and can be significantly altered even by a single outlier. To address this problem, Chazal et al.~\cite{DBLP:journals/comgeo/BuchetCOS16} introduced the 
notion of \emph{distance to measure} (DTM). 
In principle, as Chazal et al. showed, the distance to measure function can be used to compute a \Cech filtration from $P$. However in 
practice, computing the nerve of the $\alpha$-offsets requires measuring the distance at every point in the space, and so, an  
approximation to the DTM function 
is required. This is achieved by considering a finitary version of this distance, called the 
\emph{$k$-distance}, which translates to a power distance on the set of $k$-barycenters of the original point 
cloud~\cite{DBLP:journals/comgeo/BuchetCOS16,DBLP:journals/jact/AryaBDL21}. 
In general, power distances are often used to 
approximate unwieldy distance functions for computing the PH. The idea is to approximate the square of the distance to $P$ at a point 
$x\in \R^D$ by the sum of an easily computable squared distance to a point $p \in P$, together with the square of the \emph{weight} of $p$: 
$d_P(x)^2 := d'(x,p)^2 + w(p)^2$, 
where $d'(x,p)$ is chosen to be a simpler distance function, easier to compute than $d_P(x)$,
and $w(p)$ is the \emph{weight} of $p$, which is set to be a local approximation of the distance function $d_P(.)$ for points in the 
neighbourhood of $p$. In the following paragraphs, we discuss the computation of persistent homology with power distances.

        Given a set $X$ and a distance function $d:X\times X\to \R$, the pair $(X,d)$ is a \emph{metric space}
        if the distance function $d(.,.)$ is reflexive, symmetric and obeys the triangle inequality. 
        Let $\widehat{P}$ be a set of weighted points $\widehat{p}= (p,w(p))$ in a metric space $(\calM,d)$. 
        In the metric space $(\calM,d)$, the \emph{power distance} between two weighted points $\widehat{p}$ and $\widehat{q}$ is defined as 
        \[ D(\widehat{p},\widehat{q}) = d(p,q)^2 - w(p)-w(q).\]
        Accordingly, we need to extend the definition of the
        \Cech complex to sets of weighted points.

	\begin{definition}[Weighted \Cech Complex]
          Let $\wp = \{ \widehat{p}_1,...,\widehat{p}_n\}$  be a set of weighted points, where
          $\widehat{p}_i=(p_i,w_i) \in \R^D\times \R$. The
          $\alpha$-\Cech complex of $\wp$, $ \check{C}_\alpha(\wp)$,
          is  the  set of all simplices $\sigma$ satisfying
\[
         \exists x, \; \forall
           p_i\in \sigma, \; d(x,p_i)^2 \leq
           w_i+\alpha^2 \;\;\; \Leftrightarrow \;\;\; \exists x, \; \forall
           p_i\in \sigma, \; D(x,\widehat{p}_i)
           \leq \alpha^2.\]
(Here $D(x,\widehat{p}_i)$ indicates the power distance between the unweighted point $x$ (i.e. $w(x)=0$) and the weighted point $p$.)
In other words,  
 the       $\alpha$-\Cech complex of $\wp$ is  the nerve of the closed
           balls
           $\overline{B}(p_i, r_i^2=w_i+\alpha ^2)$, centered at the
           $p_i$ and of squared radius $w_i+\alpha ^2$ (if negative,
           $\overline{B}(p_i, r_i^2)$ is
           imaginary). 
         \end{definition}

        The notions of weighted \Cech filtrations and their PH now follow naturally.

	In the Euclidean case, we 
        defined the $\alpha$-\Cech complex as the collection of simplices
        whose smallest enclosing balls have radius at most
        $\alpha$. We can proceed similarly in the weighted case.
Let $\widehat{X}\subseteq \widehat{P}$. We define the  {\em radius of   $\widehat{X}$} as
\begin{eqnarray}
  \rad ^2 (\widehat{X}) &=& \min_{x\in \R^{D}} \max_{\widehat{p}_i\in \widehat{X}} D({x},\widehat{p}_i), \label{eqn:def-rad-pow-dist}
\end{eqnarray}
and the weighted center or simply the \emph{center} of $\widehat{X}$  as the point, denoted by $c (\widehat{X})$, where this
minimum is reached, i.e.  
\begin{eqnarray}
 c &=& c(\widehat{X}) = \arg\min_{x\in \R^{D}} \max_{\widehat{p}_i\in \widehat{X}} D({x},\widehat{p}_i). \label{eqn:def-cent-pow-dist}
\end{eqnarray}

        \subsection{Kernels; Gaussian Kernel Power Distance}
        \label{sec:kern-rff}
        A \emph{kernel} $K:\R^D\times \R^D \to \R$ is a similarity function on points in $\R^D$, such that 
        $K(x,x)=1$ for all $x\in \R^D$. \emph{Reproducing kernels} are a large class of kernels, having the property 
        that given a reproducing kernel $K$, there exists a lifting map $\phi$ to a Hilbert space $\calH_K$ such 
        that the kernel function lifts to the inner product on $\calH_K$, i.e. for all $x,y\in \R^D$, 
        $K(x,y) = \langle \phi(x),\phi(y)\rangle_{\calH_K}$. 

        The natural distance function induced by the norm on the Hilbert space $\calH_K$ gives a distance using the kernel 
        on $\R^D$, as follows. 
        \begin{eqnarray}
           \|\phi(x)-\phi(y)\| &=& \sqrt{\langle \phi(x)-\phi(y),\phi(x)-\phi(y)\rangle_{\calH_K}} \notag \\
                               &=& \sqrt{\langle\phi(x),\phi(x)\rangle_{\calH_K}+
                                         \langle\phi(y),\phi(y)\rangle_{\calH_K}-2\langle\phi(x),\phi(y)\rangle_{\calH_K}} \notag \\
                               &=& \sqrt{K(x,x)+K(y,y)-2K(x,y)} \;\;=\;\; \sqrt{2(1-K(x,y))}, \notag
        \end{eqnarray}
        where the last step follows since $K(x,x)=1$ for all $x\in \R^D$.
        For \emph{characteristic kernels}, a slightly smaller subset of reproducing kernels, this distance function is a metric~\cite{10.5555/1953048.2021077}.
        Our main kernel of interest in this paper is the Gaussian kernel, given by  
        \[K(x,y) = \exp\pth{-\|x-y\|^2/2\sigma^2}.\]
        For $x,y\in \R^D$, the kernel distance $D_K(.,.)$ for the Gaussian kernel is thus, 
        \begin{eqnarray}
           D_K(x,y)^2 &=& 2(1-e^{-\|x-y\|^2/2\sigma^2}). \label{eqn:dksqxy-kxy} 
        \end{eqnarray}
        

        \paragraph*{Kernel Distance to Measure}
        The definition of kernels as a similarity function on pairs of points in $\R^D$ can be extended naturally to 
        a similarity function for pairs of measures $\mu,\nu$ on $\R^D$.
        \[ \kappa(\mu,\nu) := \int_{x\in \R^D}\int_{y\in \R^D}K(x,y)d\mu(x) d\nu(y).\]
        When $\mu,\nu$ are the empirical measure on $P$ and $Q$ respectively, we get 
        \[ \kappa(\mu,\nu) := \frac{1}{|P||Q|}\sum_{p,q\in \R^D}K(p,q).\]
        Now let $ \mu $ be the empirical measure on $ P $ defined as $ \mu = \dfrac{1}{|P|}\sum_{p\in P}\delta_p$, where $ \delta_p $ 
		is the Dirac delta measure on $ P $, and $\nu$ be $\delta_p$ for a fixed point $p\in P$.
        As defined above, one can think of the \emph{kernel distance} of the point mass $\delta_p$ to the measure $\mu$, as a function of $x$, 
        which we denote by $d_{\mu}^K(x):= D_K(\mu,x)$. In~\cite{DBLP:conf/compgeom/PhillipsWZ15}, Phillips, Wang and Zheng investigated the 
        persistent homology of point sets using $d_{\mu}^K(.)$, when $K$ is a Gaussian kernel. They showed~(\cite{DBLP:conf/compgeom/PhillipsWZ15}, Theorems 4.1 and 4.2) 
        that the offsets of $P$ obtained using 
        sublevel sets of the distance function $d_{\mu}^K(x)$ in a given range of thresholds, are homotopically equivalent as long as there is no critical point of 
        $d_{\mu}^K(x)$ in this range, and are stable under perturbations of the input with bounded Hausdorff distance. 
        Thus the offsets of $d_{\mu}^K$ can be used to estimate the topological properties of the point cloud $P$.

       \paragraph*{Gaussian Kernel Power Distance}
        As mentioned in the Introduction and in Section~\ref{par:pers-homol-pow-dist}, computing the persistent homology using $d_{\mu}^K$ precisely would require computing 
        $d_{\mu}^K$ everywhere in $\R^D$. So in order to avoid this computational expense, Phillips, Wang and Zheng~\cite{DBLP:conf/compgeom/PhillipsWZ15} 
        approximated $d_{\mu}^K$ by a power distance $f_{\mu}^K(x)$, using weights on 
        the points in $P$: $f_{\mu}^K(x)^2 := \min_{p\in P} \pth{D_K^2(x,p)-w(p)}$, where $w:P\to \R$ is the \emph{Gaussian kernel weight function}  
        at $p$, defined as 
        \begin{eqnarray}
           w(p) &:=& -D_K^2(\mu,p) = -\pth{\frac{1}{|P|}\sum_{y\in P}D_K^2(p,y)- \frac{1}{2|P|^2}\sum_{x,y\in P}D_K^2(x,y)}.\label{eqn:defn-wt-func}
        \end{eqnarray}
        \begin{remark} \label{rem:gkpd} Note that the weight $w(p)$ depends not only on the point $p$, but on the pairwise squared Gaussian kernel distances  
        of the entire set $P$. This fact introduces a crucial requirement for any embedding or dimensionality reduction procedure: the kernel weight 
        function needs to be recomputed in the image space. 
        \end{remark}

        Thus from~\eqref{eqn:defn-wt-func}, the \emph{Gaussian Kernel Power Distance} (GKPD) between a point $x\in \R^D$ and the pointset $P$, can be expressed as  
        \begin{eqnarray}
              f_{\mu}^K(x)^2 &:=& \min_{p\in P} \pth{D_K^2(x,p) - w(p)}  \notag \\
                         &=&\min_{p\in P} \pth{D_K^2(x,p) +\frac{1}{|P|} \sum_{y\in P} D_K^2(y,p)- \frac{1}{2|P|^2}\sum_{y,z\in P}D_K^2(y,z)}. \label{eqn:fKP-defn}
        \end{eqnarray}
        It can be observed that the level sets of the kernel power distance are unions of balls (since the Gaussian kernel distance is a radial function of the 
        Euclidean distance). Moreover Phillips, Wang and Zheng~\cite{DBLP:conf/compgeom/PhillipsWZ15}[Theorem 3.1, Lemma 3.1] showed that for any $x\in \R^D$, 
        \[ d_{\mu}^K(x)^2\leq 2f_{\mu}^K(x)^2 \leq 4d_{\mu}^K(x)^2 + 6D_K^2(p,x).\] 
        This shows that up to constant factors, the GKPD $f_{\mu}^K$ approximates the Gaussian power distance with respect to the Gaussian kernel distance to the uniform measure,  
        $d_{\mu}^K$.

\section{Minimum Enclosing Power Balls}
\label{sec:mepb}

        As stated in the Introduction, the central technical idea behind our main result is a decomposition theorem for the squared radius of the minimum enclosing ball of 
        a set of weighted points, computed using the Gaussian kernel power distance. To prove such a decomposition theorem, we need to understand some properties of 
        minimum enclosing balls of collections of weighted points under power distances. However, since the Gaussian kernel distance is non-linear, 
        in order to prove our new results on the minimum enclosing balls of weighted points 
        under the GKPD, it will be crucial to lift the points onto a Hilbert space where the Gaussian kernel distance corresponds to the norm of an inner product. 
        In this section, we shall therefore investigate such minimum enclosing balls in a Hilbert space with a norm given by an inner product and prove some new 
        properties which will be crucial in the proof of our technical result in the next section. Note that since a Euclidean space is also a Hilbert space and the usual 
        $\ell_2$ distance appears as the norm of an inner product, all our results also apply in the usual Euclidean setting.\\

        Recall the definitions of the weighted \Cech complex 
        using a power distance from Section~\ref{par:pers-homol-pow-dist}. 
        We are given a set $\{\widehat{p}_1,\ldots,\widehat{p}_k\}$ of weighted points in a Hilbert space $\calH$, where for all $i\in [k]$, $\widehat{p}_i := (p_i,w(p_i))$ 
        denotes the point $p_i\in \calH$, having weight $w(p_i)$. The power distance between a pair of points $x,y\in \calH$ is given by 
        \[ D(\hat{x},\hat{y}) := \|x-y\|_{\calH}^2 - w(x)-w(y),\]
        where $\|.\|_{\calH}$ is the norm induced by the inner product of the Hilbert space $\calH$ and $w(x)=0$ if $x\not\in P$.
        Let $\widehat{\sigma}$ denote the simplex formed by the points. Using $D(\hat{x},\hat{y})$ defined above, we define 
        the center $c(\widehat{\sigma})$ and radius $\rad(\widehat{\sigma})$ as in equations~\eqref{eqn:def-cent-pow-dist} and~\eqref{eqn:def-rad-pow-dist} 
        respectively.

        \begin{lemma} 
        \label{l:uniq-cent-rad} 
             The center $c(\widehat{\sigma})$ and radius $\rad(\widehat{\sigma})$ are unique.
        \end{lemma}

        \begin{proof}[Proof of Lemma~\ref{l:uniq-cent-rad}]
             Suppose there exist distinct centers $c_0\neq c_1$. Let $r$ denote $\rad(\widehat{\sigma})$. By the definition of center, we have 
             \[ \exists \widehat{p}_0 \in \widehat{X}:\;\; \forall \widehat{p}_i \in \widehat{X}:\;\; D(c_0,\widehat{p}_i) \leq D(c_0,\widehat{p}_0) = \|c_0-p_0\|_{\calH}^2 - w(p_0) = r^2.\]
             \[ \exists \widehat{p}_1 \in \widehat{X}:\;\; \forall \widehat{p}_i \in \widehat{X}:\;\; D(c_1,\widehat{p}_i) \leq D(c_1,\widehat{p}_1) = \|c_1-p_1\|_{\calH}^2 - w(p_1) = r^2.\]
         For any $\lambda\in (0,1)$, define 
         \[D_{\lambda}(\widehat{p}_i):= (1-\lambda)D(c_0,\widehat{p}_i)+\lambda D(c_1,\widehat{p}_i),\] and 
         \[c_{\lambda} := (1-\lambda)c_0+\lambda c_1.\] 
         Then,
         \begin{eqnarray}
            D_{\lambda}(\widehat{p}_i) &=& (1-\lambda)D(c_0,\widehat{p}_i)+\lambda D(c_1,\widehat{p}_i) \notag \\
                                   &=& (1-\lambda)\|c_0-p_i\|_{\calH}^2-(1-\lambda)w(p_i)+\lambda \|c_1-p_i\|_{\calH}^2 -\lambda w(p_i) \notag\\
                                   &=& (1-\lambda)\|c_0-p_i\|_{\calH}^2+\lambda \|c_1-p_i\|_{\calH}^2 - w(p_i) \notag\\
                                   &>& \|(1-\lambda)c_0+\lambda c_1-p_i\|_{\calH}^2 - w(p_i) \label{eqn:conv-sq-norm}\\
                                   &=& D(c_{\lambda},p_i) \notag,
         \end{eqnarray}
         where step~\eqref{eqn:conv-sq-norm} followed from the fact that a distance function defined using a norm is a strictly convex function.
         The last line above contradicts our assumption that $c_0$ and $c_1$ are distinct centers of $\widehat{X}$, and so we get that $\widehat{X}$ can 
         have only one center.  
        \end{proof}

         Let $I$ be the set of indices of $p_j \in \hat{X}$, such that $\rad^2(\widehat{X}) = D(c,\widehat{p}_j)$, and let $\hat{X}_I$ be 
         the corresponding set of weighted points. The next lemma shows that the center $c$ can be expressed as a convex combination of the points of $\hat{X}_I$. 
        \begin{lemma} 
        \label{l:cent-conv-sum} 
         There exists a collection of non-negative
         real numbers $(\lambda_i)_{i\in I}$ such that $\sum_{i\in I} \lambda_i = 1$ and $c = \sum_{i\in I} \lambda_i p_i$.
        \end{lemma}

        \begin{proof}[Proof of Lemma~\ref{l:cent-conv-sum}]
         Suppose $c\not\in \conv(X_I)$. By the Hilbert Projection Theorem~\cite{10.5555/26851}, there exists a unique $c'\neq c$ such that $c' = \arg\inf_{x\in \conv(X_I)} \|c-x\|$.
         Let $\tilde{c} = \lambda c+(1-\lambda)c'$, for some $\lambda\in [0,1]$. Then for any $p_i$, $i\in I$, the distance $\|\tilde{c}-p_i\|$ satisfies  
         \begin{eqnarray*}
            \|\tilde{c}-p_i\|_{\calH} &=& \|\lambda c+(1-\lambda)c'-p_i\|_{\calH} \\
                              &=& \|\lambda (c-p_i)+(1-\lambda)(c'-p_i) \|_{\calH} \\
                              &\leq& \lambda \|c-p_i\|_{\calH}+(1-\lambda)\|c'-p_i \|_{\calH} \;\;\leq\;\; \|c-p_i\|_{\calH},
         \end{eqnarray*}
         where the last step follows since $\|c'-p_i\|_{\calH}\leq \|c-p_i\|_{\calH}$ for any $i\in I$.
         By squaring the distances and then subtracting the weight of $p_i$ from both sides, we get that for all $i\in I$,  
         \[ D(\tilde{c},\hat{p}_i) \leq D(c,\hat{p}_i).\]
         For $\lambda$ sufficiently close to $1$, $\tilde{c}$ remains closer to the weighted points $\hat{p}_k$, $k\not\in I$ than to $\hat{p}_i$, $i\in I$.
         Thus we get 
         \[ D(\tilde{c},\hat{p}_k)<D(\tilde{c},\hat{p}_i)<D(c,\hat{p}_i) = \rad^2(X_I),\]
         which contradicts the fact that $c$ is the center of $X_I$.

        \end{proof}

        The next lemma shows that a convex combination of the squared distances of the points of $\hat{X}_I$ from the center $c$ can be  
        expressed in terms of a linear combination of pairwise distances between the points of $\hat{X}_I$. The proof of this lemma is 
        where we require the assumption that the points lie in a Hilbert space $\calH$ with norm given by an inner product. 

        \begin{lemma}
        \label{l:var-sq-ex-genl}
          Let $c = \sum_{i\in I} \lambda_i p_i$, where $\lambda_i\geq 0$ for $i\in I$, and $\sum_{i\in I}\lambda_i =1$. Then 
          \begin{eqnarray}
             \sum_{i\in I} \lambda_i \|c-p_i\|_{\calH}^2 &=& \frac{1}{2}\sum_{i,j\in I} \lambda_i\lambda_j \|p_i-p_j\|_{\calH}^2.
          \end{eqnarray}
        \end{lemma}

        \begin{proof}[Proof of Lemma~\ref{l:var-sq-ex-genl}] 
          While the Lemma can be proved using basic linear algebra, we shall state the proof in probabilistic language, 
          which, in our opinion, makes it simpler and more intuitive.
          Let $X_1,X_2$ be two independently random points chosen from $\{p_i\;|\; i\in I\}$, with the point $p_i$ 
          being chosen with probability $\lambda_i$, for each $i\in I$. 
          By definition, we have $\|X_1-X_2\|_{\calH}^2 = \langle X_1-X_2,X_1-X_2\rangle$.
          Using the linearity of expectation we get 
          \begin{eqnarray*}
              \Ex{\|X_1-X_2\|_{\calH}^2} &=& \Ex{\langle X_1-X_2,X_1-X_2\rangle} \\
                                &=& \Ex{\langle X_1,X_1\rangle} + \Ex{\langle X_2,X_2\rangle} - \Ex{\langle X_2,X_1\rangle} - \Ex{\langle X_1,X_2\rangle} \\
                                &=& 2\pth{\Ex{\langle X_1,X_1\rangle} -\Ex{\langle X_2,X_1\rangle}}. 
          \end{eqnarray*}
          In the above, the last step followed from the fact that $X_1$ and $X_2$ are independent and equidistributed random variables. 
          The term $\Ex{\langle X_2,X_1\rangle}$ can be evaluated by first taking the expectation over 
          $X_2$, using the linearity of the inner product over the first variable, and then taking expectation over $X_1$. Thus we get 
          \[ \Ex{\langle X_2,X_1\rangle} = \Ex{\langle \Ex{X_2},X_1\rangle} = \Ex{\langle \Ex{X_1},X_1\rangle} = \langle \Ex{X_1},\Ex{X_1}\rangle.\]
          Therefore, the expression for $\Ex{\|X_1-X_2\|_{\calH}^2}$ can be written as 
          \begin{eqnarray*}
              \Ex{\|X_1-X_2\|_{\calH}^2} &=& 2\pth{\Ex{\langle X_1,X_1\rangle} -\langle \Ex{X_1},\Ex{X_1}\rangle} \\
                                &=& 2\pth{\Ex{\langle X_1-\Ex{X_1},X_1-\Ex{X_1}\rangle} } \;\;=\;\; 2\Ex{\|X_1-\Ex{X_1}\|_{\calH}^2}. 
          \end{eqnarray*}
          Evaluating the expectations, we get $\Ex{X_1} = \sum_{i\in I} \lambda_i p_i = c$, so that  
          \[\Ex{\|X_1-X_2\|_{\calH}^2} = \sum_{i,j\in I}\lambda_i\lambda_j \|p_i-p_j\|_{\calH}^2,\] 
          and 
          \[\Ex{\|X_1-\Ex{X_1}\|_{\calH}^2} = \sum_{i\in I}\lambda_i \|p_i-c\|_{\calH}^2.\] 
          Thus, the final expression becomes
          \begin{eqnarray*}
              \sum_{i,j\in I}\lambda_i\lambda_j \|p_i-p_j\|_{\calH}^2 &=& \Ex{\|X_1-X_2\|_{\calH}^2} \\ 
                                                             &=& 2\Ex{\|X_1-\Ex{X_1}\|_{\calH}^2}  
                                                             \;\;=\;\; 2 \sum_{i\in I}\lambda_i \|c-p_i\|_{\calH}^2.
          \end{eqnarray*}
        \end{proof} 

        The following Decomposition Lemma shows that the squared radius of the minimum enclosing 
        ball of the weighted point set $\widehat{X}$ can be expressed as a combination of pairwise power distances of the points in 
        $\widehat{X}$.

        \begin{lemma}[Decomposition Lemma] 
        \label{l:rad-conv-sum} 
          Let $I$ be the set of indices as defined in Lemma~\ref{l:cent-conv-sum}. Then 
          \[ \rad^2(\widehat{X}) = \frac{1}{2}\sum_{i\in I}\sum_{j\in I} \lambda_i\lambda_j D(\widehat{p}_i,\widehat{p}_j).\]
        \end{lemma}
 
        \begin{proof}[Proof of Lemma~\ref{l:rad-conv-sum}]
          For all $i\in I$, we have 
          \begin{eqnarray}
             \rad^2(\widehat{X}) &=& \sum_{i\in I}\lambda_i (\|c-p_i\|_{\calH}^2 -w(p_i)),\label{eqn:rad-conv-comb-wtd-dist-cent}
          \end{eqnarray}
          so that by Lemma~\ref{l:var-sq-ex-genl}, we have
          \begin{eqnarray*}
             \sum_{i\in I} \lambda_i \|c-p_i\|_{\calH}^2 &=& \frac{1}{2}\sum_{i,j\in I} \lambda_i\lambda_j \|p_i-p_j\|_{\calH}^2.
          \end{eqnarray*}
          Substituting in equation~\eqref{eqn:rad-conv-comb-wtd-dist-cent}, 
          \begin{eqnarray*}
              \rad^2(\widehat{X}) &=& \frac{1}{2}\sum_{i,j\in I}\lambda_i\lambda_j \|p_i-p_j\|_{\calH}^2 -\frac{1}{2}\sum_{i\in I} 2\lambda_i w(p_i) 
                              \;\;=\;\; \frac{1}{2}\sum_{i,j\in I}\lambda_i\lambda_j \|p_i-p_j\|_{\calH}^2 -\frac{1}{2}\sum_{i,j\in I} 2\lambda_i\lambda_j w(p_i) \\
                              &=& \frac{1}{2}\sum_{i,j\in I}\lambda_i\lambda_j \|p_i-p_j\|_{\calH}^2 -\frac{1}{2}\sum_{i,j\in I} \lambda_i\lambda_j (w(p_i)+w(p_j)) \\
                              &=& \frac{1}{2}\sum_{i,j\in I}\lambda_i\lambda_j \pth{\|p_i-p_j\|_{\calH}^2 - w(p_i)-w(p_j)} 
                              \;\;=\;\; \frac{1}{2}\sum_{i,j\in I}\lambda_i\lambda_j D(\widehat{p}_i,\widehat{p}_j). 
          \end{eqnarray*}
        \end{proof}

        To end this section, we mention the following interesting corollaries for the GKPD. Consider a 
        set of points $p_1,\ldots,p_k \in \R^D$ weighted using the GKPD as defined in~\eqref{eqn:defn-wt-func}, and 
        let $\widehat{\sigma}$ denote the associated abstract simplex formed by $\{\widehat{p}_1,\ldots,\widehat{p}_k\}$. 
        From Section~\ref{subsec:e-dist-maps-rad-simp-gkpd}, there exists a lifting map $\phi:\R^D\to \calH_K$ where $\calH_K$ is a Hilbert space, such 
        that the squared Gaussian kernel distance between any pair of points corresponds to the norm of the inner product of 
        ther difference of their position vectors in $\calH_K$. 
        By applying Lemma~\ref{l:uniq-cent-rad} and Lemma~\ref{l:rad-conv-sum} to the images of $\hat{p}_i$ $i\in [k]$ under the lifting map $\phi$ and using the corresponding 
        statements for the pairwise distances of the original points under the GKPD, we derive the following new properties of minimum enclosing balls under the GKPD: 
        \begin{corollary}
        \label{cor:new-props-gkpd}
        \begin{enumerate} 
           \item The center and radius of $\widehat{\sigma}$ are unique.
           \item There exists a set of non-negative reals $(\lambda_i)_{i\in [k]}$, such that $\sum_{i\in [k]}\lambda_i=1$, and 
          \[ \rad^2(\widehat{\sigma}) = \frac{1}{2}\sum_{i\in [k]}\sum_{j\in [k]} \lambda_i\lambda_j D(\widehat{p}_i,\widehat{p}_j).\]
        \end{enumerate} 
        \end{corollary}

\section{Low-Distortion Maps for Power Distances}
\label{sec:low-dist-map-pd}
        In this section, we shall look at low-distortion mappings of power distances.
        First we need the notion of an \emph{$\e$-distortion} map for power distances between metric spaces.
        \begin{definition}
            Given metric spaces $(X,d_X)$ and $(Y,d_Y)$, 
        a point set $P\subset X$ and $\e\in (0,1)$, 
        a mapping $f:X\to Y$ is an \emph{$\e$-distortion map} with respect to pairwise distances in $P$ if, 
        \[  \forall x,y\in P:\;\; (1-\e) d_X(x,y)^2\leq d_Y(f(x),f(y))^2 \leq (1+\e)d_X(x,y)^2.\]
        Further, given a pair of weight functions $w_X:P \to \R$ and $w_Y:P\to \R$, 
        $f$ is an $\e$-distortion map with respect to $w_X$ if, 
        \[ \forall x\in P:\;\;  |w_Y(f(x))-w_X(x)| \leq \e w_X(x) .\]
        Thus, given a power distance defined as $D_X(\hat{x},\hat{y}) = d_X(x,y)^2 -w_X(x)-w_X(y)$, where $x,y\in P$, and $\hat{x} = (x,w_X(x))$, $x\in P$, are 
        weighted points, the mapping $f:X\to Y$ is an \emph{$\e$-distortion mapping for the power distance} $D_X$ if $f$ is an $\e$-distortion 
        mapping for the pairwise distances as well as for the weight functions $w_X,w_Y$. 
        \end{definition}

        \subsection{Reproducing Kernel Hilbert Spaces and Dimensionality Reduction} 
        In the context of Gaussian kernels, perhaps the most well-known example of an $\e$-distortion map is the RFF 
        map of Rahimi and Recht~\cite{DBLP:conf/nips/RahimiR07}, which was shown 
        by Chen and Phillips~\cite{DBLP:conf/alt/ChenP17} to be an $\e$-distortion map for the Gaussian kernel distance. \\

        For points in $\R^D$, there exists a mapping to $\R^t$, with $t= O\pth{\e^{-2}\log n}$ that 
        gives a relative approximation of the kernel distance on $\R^D$, as the natural inner product on $\R^t$. 
        Their RFF mapping is given as follows:
        For $i=1,\ldots,d/2$, given $\sigma \geq 0$, let $\omega_i \sim \calN_D(0,\sigma^{-2})$ be 
        $D$-dimensional independent Gaussian random variables. For each $i$, define the random map $f_i:\R^D\to \R^2$, as 
        \[ f_i(x) = (\cos(\langle \omega_i,x\rangle),\sin(\langle\omega_i,x\rangle)).\] 
        Finally, define the mapping $f:\R^D\to \R^t$ as   
        \begin{eqnarray}
             f(x) &=& \bigotimes_{i=1}^{d/2} f_i(x). \label{eqn:def-rff-chph}
        \end{eqnarray}

        \begin{theorem}[Chen, Phillips~\cite{DBLP:conf/alt/ChenP17}] 
        \label{thm:chen-phillips-dimred}
        Given any $\e,\delta\in (0,1)$, for any set $P\subset \R^D$ 
            of $(i)$ $n$ points, or $(ii)$ an arbitrary number of points, such that for all $x,y\in P$, $\|x-y\|/\sigma \leq r$, where $r>0$ is a given parameter, 
        $f:\R^D\to \R^t$ defined as in eqn.~\eqref{eqn:def-rff-chph}, with $(i)$ $t:=  \Omega\pth{\e^{-2}\log (n/\delta)}$ dimensions, or 
        $(ii)$ $t:= \Omega\pth{\e^{-2}D\log(rD/\e\delta)}$ dimensions respectively, is an $\e$-distortion map for the Gaussian kernel distance, 
        i.e. $f$ satisfies $\frac{\|f(x)-f(y)\|^2}{D_K^2(x,y)} \in (1-\e,1+\e)$ 
        for all pairs of points $x,y\in P$, with probability at least $1-\delta$.
        \end{theorem}
        
        \begin{remark} 
        For the Gaussian kernel function, there exists a lifting map $\phi$ from $\R^D$ to a Hilbert space $\calH_K$ such that the kernel
        function lifts to the inner product on $\calH_K$ (for any pair of points in $\R^D$).
        The Hilbert space $\calH_K$ corresponding to a reproducing kernel $K(.,.)$ 
        is in general infinite-dimensional. However,  using the RFF maps of Rahimi and 
        Recht~\cite{DBLP:conf/nips/RahimiR07}, Chen and Phillips~\cite{DBLP:conf/alt/ChenP17} 
        showed that the inner product on $\calH_K$ can be approximated by the Euclidean inner product on a finite-dimensional space. 
        \end{remark}

            From Remark~\ref{rem:gkpd}, we know that the Gaussian weight function needs to be recomputed in the image space. For each $p\in P$, let us define 
        \begin{eqnarray}
            w_{t}(p) := -\pth{\frac{1}{|P|}\sum_{y\in P}\|f(p)-f(y)\|^2- \frac{1}{2|P|^2}\sum_{x,y\in P}\|f(x)-f(y)\|^2}. \label{eqn:recomp-wts-rt}
        \end{eqnarray}
        Now comparing the definition of the Gaussian kernel weight function eq.~\eqref{eqn:defn-wt-func} with eq.~\eqref{eqn:recomp-wts-rt}, 
        and applying Theorem~\ref{thm:chen-phillips-dimred}, 
        it can be seen that $f$ is an $\e$-distortion map for the Gaussian kernel weight function as well, and therefore an $\e$-distortion map for the GKPD. 

        \begin{corollary}
        \label{cor:eps-distort-wts}
            The map $f$ of Theorem~\ref{thm:chen-phillips-dimred} is an $\e$-distortion map for the GKPD.
        \end{corollary}

        \subsection{$\e$-Distortion Maps and Radii of Simplices under GKPD}
        \label{subsec:e-dist-maps-rad-simp-gkpd}
        Now we are ready to show that 
        applying an $\e$-distortion map to a point sample 
        in a high-dimensional space can change the 
        radii of the minimum enclosing balls of the simplices of the \Cech filtration at most by a $(1\pm\e)$-factor. 
        This is the Simplex Distortion Lemma, stated and proved below.


        \begin{lemma}[Simplex Distortion Lemma]
        \label{l:rad-eps-distortion}
           Let $\widehat{\sigma} \subset \widehat{P}$ be a simplex in the weighted \Cech complex $\check{C}_{\alpha}(\widehat{P})$ using the Gaussian 
           kernel power distance $f_{\mu}^K:\R^D\to \R$, and let 
           $G:(\R^D,D_K)\to (\R^t,\|.\|)$ be an $\e$-distortion map for the pairwise power distances given by $D(\hat{p},\hat{q})= D_K^2(p,q)-w(p)-w(q)$, 
           $\hat{p},\hat{q}\in \widehat{P}$. 
           Then  
         \[ (1-\e)\rad^2(\widehat{\sigma}) \leq \rad^2(\widehat{G(\sigma)}) \leq (1+\e)\rad^2(\widehat{\sigma}) ,\]
           where $\widehat{G(\sigma)}$ denotes the image of the simplex $\sigma$ in $\R^t$, with the  
           weights being recomputed in $\R^t$ as in~\eqref{eqn:recomp-wts-rt}.
        \end{lemma}


        \begin{proof}[Proof of Simplex Distortion Lemma~\ref{l:rad-eps-distortion}]
           Let the simplex $\widehat{\sigma} = \{\widehat{p}_1,\ldots,\widehat{p}_k\}$, where for all $i\in [k]$, $\widehat{p}_i := (p_i,w(p_i))$ is a 
         weighted point, and let $c(\widehat{\sigma})$ and $\rad(\widehat{\sigma})$ denote its center and radius respectively. 

         Since the Gaussian kernel is a characteristic kernel, there exists a Hilbert space $\calH_K$, and a lifting map $\phi:\R^D \to \calH_K$, such that 
         for all $x,y\in \R^D$, $D_K(x,y) = \|\phi(x)-\phi(y)\|_{\calH_K}$ (see e.g.~\cite{DBLP:conf/alt/ChenP17,DBLP:conf/approx/PhillipsT20}). 
         By eqn.~\eqref{eqn:defn-wt-func} the weight of a 
         point $p\in P$ is a weighted sum of squared kernel distances: 
           \[ w(p) := -D_K^2(\mu,p) = -\pth{\frac{1}{|P|}\sum_{y\in P}D_K^2(p,y)- \frac{1}{2|P|^2}\sum_{x,y\in P}D_K^2(x,y)}.\]
         Thus the lifting map $\phi$ extends naturally to the weights $w(p_i)$, $i\in [k]$, as 
         \[ \phi(w(p_i)) = -\pth{\frac{1}{|P|}\sum_{y\in P}\|\phi(p)-\phi(y)\|_{\calH_K}^2- \frac{1}{2|P|^2}\sum_{x,y\in P}\|\phi(x)-\phi(y)\|_{\calH_K}^2},\]
         which allows us to define the weights in the lifted space as $w(\phi(p)) := \phi(w(p))$.

         Applying the Decomposition 
         Lemma~\ref{l:rad-conv-sum} with $\calH = \calH_K$, and the lifted weighted points given by $\phi(\hat{p}_i):=(\phi(p_i),\phi(w(p_i)))$, we have 
         \begin{eqnarray}
            \rad^2(\widehat{\sigma}) &=& \frac{1}{2}\sum_{i,j\in [k]} \lambda_i\lambda_j D(\widehat{p_i},\widehat{p_j}).
         \end{eqnarray}
         Since $G$ is an $\e$-distortion map, for each pair $\widehat{p_i},\widehat{p_j} \in \widehat{\sigma}$, we have 
         \begin{eqnarray}
            (1-\e)\|\phi(p_i)-\phi(p_j)\|_{\calH_K}^2 &\leq& \|G(p_i)-G(p_j)\|^2 \;\;\leq \;\; (1+\e)\|\phi(p_i)-\phi(p_j)\|_{\calH_K}^2
         \end{eqnarray}
         Since $G$ is an $\e$-distortion map for the weight function as well, we get for each $\widehat{p_i}\in \widehat{\sigma}$,
         \begin{eqnarray}
            (1-\e)w(\widehat{p_i}) &\leq& w(\widehat{G(p_i)}) \;\;\leq \;\; (1+\e)w(\widehat{p_i}).
         \end{eqnarray}
         Subtracting the weights $w(\widehat{G(p_i)}), w(\widehat{G(p_j)})$ from the squared distance $\|G(p_i)-G(p_j)\|^2$, and using that 
         $D(\hat{p_i},\hat{p_j}) = D_K^2(p_i,p_j)-w(p_i)-w(p_j) = \|\phi(p_i)-\phi(p_j)\|_{\calH_K}-w(p_i)-w(p_j)$, we get 
         \begin{eqnarray}
            (1-\e)D(\widehat{p_i},\widehat{p_j}) &\leq& D(\widehat{G(p_i)},\widehat{G(p_j)})  \;\;\leq \;\; (1+\e)D(\widehat{p_i},\widehat{p_j}).
         \end{eqnarray}
         Let $\widehat{G(\sigma)}$ denote the image of the simplex $\widehat{\sigma}$  under the map $G$, and $c(\widehat{G(\sigma)})$ be its center. 
         Applying Lemma~\ref{l:cent-conv-sum} on the space $(\R^t,\|.\|)$, we get that $c(\widehat{G(\sigma)})$ is a convex combination of the vertices of $\widehat{G(\sigma)}$, say 
         \[ c(\widehat{G(\sigma)})= \sum_{i\in [k]} \mu_i G(\widehat{p_i}),  \] 
         where $\forall i\in [k]; \;\; \mu_i \geq 0$ and $\sum_{i\in [k]} \mu_i = 1$\footnote{Note that the convex combination $(\mu_i)_{i\in [k]}$ need not 
         be the same as the combination $(\lambda_i)_{i\in [k]}$ for $c(\widehat{\sigma}))$ in the original space.}.
         Since $G$ is an $\e$-distortion map, using the Decomposition Lemma~\ref{l:rad-conv-sum} we get 
         \begin{eqnarray}
            \rad^2(\widehat{G(\sigma)})&=& 
                             \sum_{i,j\in [k]}\mu_i\mu_j\pth{\frac{1}{2}D(\widehat{G(p_i)},\widehat{G(p_j)})}  \notag \\
             &\geq& \sum_{i,j\in [k]}\frac{\mu_i\mu_j}{2}\pth{(1-\e)D(\widehat{p_i},\widehat{p_j})}  
         \end{eqnarray}
         \begin{eqnarray}
            \mbox{i.e. } \sum_{i,j\in [k]}\frac{\mu_i\mu_j}{2}\pth{D(\widehat{p_i},\widehat{p_j})} 
                    &\leq& \frac{\rad^2(\widehat{G(\sigma)})}{1-\e}.  
                                                \label{eqn:mus-radg}
         \end{eqnarray}
         Also, by the minimality in the definition of the squared radius of a weighted simplex, we have  
         \begin{eqnarray}
             \rad^2(\widehat{\sigma}) &=& \frac{1}{2}\sum_{i,j\in [k]}\lambda_i\lambda_j D(\widehat{p_i},\widehat{p_j}) \;\;\leq\;\;
             \frac{1}{2}\sum_{i,j\in [k]}\mu_i\mu_j D(\widehat{p_i},\widehat{p_j}),  \text{\hspace{2mm} and}
                                                \label{eqn:rad-mus}
         \end{eqnarray}
         \begin{eqnarray}
             \rad^2(\widehat{G(\sigma)}) &=& \frac{1}{2}\sum_{i,j\in [k]}\mu_i\mu_j D(\widehat{G(p_i)},\widehat{G(p_j)})  
                                             \;\;\leq\;\;
                                             \frac{1}{2}\sum_{i,j\in [k]}\lambda_i\lambda_j D(\widehat{G(p_i)},\widehat{G(p_j)}), \notag \\  
                                         &\leq& 
                                             \sum_{i,j\in [k]} \frac{\lambda_i\lambda_j}{2}\pth{(1+\e)D(\widehat{p_i},\widehat{p_j}) }, 
                                                \label{eqn:radg-lambdas} \\
                                         &=& 
                                             (1+\e)\rad^2(\widehat{\sigma}) 
                                                \label{eqn:radg-rad} 
         \end{eqnarray}
         where in step~\eqref{eqn:radg-lambdas} we again used that $G$ is an $\e$-distortion map, 
         and the last step followed from the Decomposition Lemma~\ref{l:rad-conv-sum}. 
         Combining equations~\eqref{eqn:mus-radg}, ~\eqref{eqn:rad-mus} and~\eqref{eqn:radg-rad} gives 
         \begin{eqnarray}
            (1-\e)\rad^2(\widehat{\sigma}) &\leq& \rad^2(\widehat{G(\sigma)}) \;\;\leq\;\; (1+\e)\rad^2(\widehat{\sigma}),
         \end{eqnarray}
         which completes the proof of the lemma.
        \end{proof}

        We conclude the section by stating the main decomposition theorem.
        \begin{theorem}
        \label{thm:inner-prod-dimred-gauss-kern}
          Let $\hat{P}=\{\hat{p}_1,\ldots,\hat{p}_k\}$ be a set of weighted points such that $\hat{p}=(p,w(p))$, with $p\in \R^D$ and $w:\R^D\to\R$ being the weight 
          function given by the GKPD (equation~\eqref{eqn:defn-wt-func}).
          Let $G:(\R^D,D_K)\to (\R^t,\|.\|)$ be an 
          $\e$-distortion map for the power distance given by $D(\hat{p}_i,\hat{p}_j)=D_K^2(p_i,p_j) - w(p_i)-w(p_j)$. Then the \Cech 
          filtration $\check{C}_{\alpha}(\widehat{P})$ computed using the GKPD $f_{\mu}^K:\R^D\to \R$, and the corresponding weighted 
          filtration in $\R^t$ under the map $G$, (with the weights being recomputed in $\R^t$ as in~\eqref{eqn:recomp-wts-rt}), 
          are multiplicatively $((1-\e)^{-1}+o(1))$-interleaved. 
        \end{theorem}

        The proof of Theorem~\ref{thm:inner-prod-dimred-gauss-kern} 
        follows directly by applying the Simplex Distortion Lemma to each simplex in $\check{C}_{\alpha}(\widehat{P})$. 

	\section{Proof of Main Result}
	\label{sec:main-result-proof}

        In this section, we shall prove our main result, which is stated below. 

        \begin{theorem}
        \label{thm:gauss-kern-dimred} 
        Given $\sigma>0$, $\e,\delta \in (0,1)$, $\sigma>0$, and 
        a finite set $P \subset \R^D$ consisting of 
        \begin{enumerate} 
            \item[(i)] $n$ points, or 
            \item[(ii)] an arbitrary number of points having Euclidean diameter at most $r\sigma$, 
                        where $r>0$ is a given parameter, 
        \end{enumerate} 
        a Random Fourier Features projection map onto $(i)$ $t := \Omega\pth{\e^{-2}\log (n/\delta)}$ dimensions,   
        or $(ii)$ $t := \Omega\pth{D\e^{-2}\log (Dr/\e \delta)}$ dimensions respectively, is such that 
        the \Cech filtration 
        computed using the GKPD 
        and the corresponding weighted filtration in $\R^t$, are $((1-\e)^{-1}+o(1))$-interleaved 
        with probability at least $1-\delta$. Further, the Random Fourier Features map $f:\R^D\to \R^t$ can be computed in time $O\pth{nt}$.
        \end{theorem}


        \begin{proof}[Proof of Theorem~\ref{thm:gauss-kern-dimred}] 

           In order to prove that an RFF mapping onto $t$ dimensions gives a data set whose weighted \Cech filtration (with the weights being 
        recomputed in the image space) is interleaved with the original filtration, it suffices to show that for an arbitrary weighted simplex $\sigma$ 
        the radius of $\sigma$ under the GKPD is distorted at most by a $(1\pm \e)$-factor under the RFF mapping, with high probability. 
        Thus, our overall plan shall be to apply Theorem~\ref{thm:inner-prod-dimred-gauss-kern}, using the $\e$-distortion map obtained 
        from Theorem~\ref{thm:chen-phillips-dimred}. The following facts, proved in Section~\ref{sec:mepb} -- especially Corollary~\ref{cor:new-props-gkpd}  --  
        and Section~\ref{sec:low-dist-map-pd} shall be central to the proof.

        \begin{enumerate}
           \item The squared radius of the original weighted simplices in $\check{C}_{\alpha}(\widehat{P})$ under the GKPD, in $\R^D$, is computable as a linear 
                 combination of the pairwise GKPDs of the vertices of $\hat{\sigma}$.
           \item The squared radius of the weighted simplices in $\check{C}_{\alpha}(\widehat{G(P)})$ under the Euclidean distance and weights being recomputed in 
                 $\R^t$ (as in~\eqref{eqn:recomp-wts-rt}, is computable as a linear combination of the weighted Euclidean distances of the vertices of $\widehat{G(\sigma)}$.  
           \item The RFF mapping approximately preserves the pairwise GKPDs between the vertices of $\hat{\sigma}$, and hence approximately preserves 
                 \emph{every} linear combination of these power distances, and in particular, the linear combinations corresponding to the original and the new radii 
                 of $\widehat{\sigma}$, which can now be compared.
        \end{enumerate}

        To apply Theorem~\ref{thm:inner-prod-dimred-gauss-kern} therefore, let us verify that given an $\e$-distortion map, the 
        conditions for the theorem hold. 
        Applying Theorem~\ref{thm:chen-phillips-dimred}, we get the function $f:\R^D\to \R^t$ such that  
        \[ \forall x,y\in P:\;\; (1-\e)D_K(x,y) \leq \|f(x)-f(y)\| \leq (1+\e)D_K(x,y).\]
        Thus we get that $f$ is an $\e$-distortion map for the kernel distance $D_K$. 
        For the weight function $w:P\to \R$, observe from eqn.~\eqref{eqn:defn-wt-func} that the weight of a point $p\in P$ is a weighted sum of squared kernel distances. 
        Now with the weights in $\R^t$ being redefined as in eqn.~\eqref{eqn:recomp-wts-rt}, Corollary~\ref{cor:eps-distort-wts} shows that the weights 
        are $(1\pm\e)$-preserved. 

        Therefore, all the conditions of Theorem~\ref{thm:inner-prod-dimred-gauss-kern}
        are satisfied. 


        Now we can apply Theorem~\ref{thm:inner-prod-dimred-gauss-kern} to get that the weighted \Cech filtration $\check{C}_{\alpha}(\widehat{P})$ built 
        using the kernel distance function, interleaves with the \Cech filtration built 
        on the image of the weighted point set $\widehat{P}$ under $G$ using the Euclidean distance, i.e. $\check{C}_{\alpha}(\widehat{G(P)})$, as follows,
          \[ \check{C}_{\alpha_-}(\widehat{P})\subseteq  \check{C}_{\alpha}(\widehat{G(P)})\subseteq \check{C}_{\alpha_+}(\widehat{P}).\]
          where $\alpha_- := \alpha\sqrt{1-\e}$ and $\alpha_+ := \alpha\sqrt{1+\e}$.
        Finally, since for $\e \in (0,1)$ we have $1<(1+\e)<(1-\e)^{-1}$, we can apply the definition of $\beta$-interleaved in Section~\ref{sec:persh} 
        to get the statement of the Theorem.
        \end{proof}



        \section{Conclusion}
        \label{sec:concl}
           We have shown that the Random Fourier Features map can be used to reduce the dimensionality of input data, for building persistence diagrams. 
        It should be noted that the reduced dimension -- similar to applications of the Johnson-Lindenstrauss Lemma -- has an $\e^{-2}$ factor, and therefore cannot be really 
        small if the allowed error $\e$ is very small. Still it can reduce the dimensionality from a few millions (typical in AI applications) to a few hundreds. It would be 
        interesting if this technique could be combined with other dimensionality reduction techniques such as PCA or gradient descent based techniques.  \\

\bibliographystyle{plainurl}
	\bibliography{bibliography} 

\end{document}